\documentclass[sigconf,screen,nonacm]{acmart}

\usepackage{booktabs} 

\setcopyright{none}  

\AtEndPreamble{
  \newtheorem{observation}[theorem]{Observation}
  \newtheorem{remark}[theorem]{Remark}
  \newtheorem{assumption}[theorem]{Assumption}

}

\usepackage{tikz}
\usetikzlibrary{positioning,shapes,shadows,arrows,calc}

\usepackage{xspace}
\usepackage{bm}

\usepackage[capitalize]{cleveref}
\crefname{equation}{eq.}{equations}
\crefname{definition}{def.}{definitions}
\crefname{section}{sec.}{sections}
\crefname{lemma}{lem.}{lemmata}
\crefname{observation}{obs.}{observations}
\crefname{corollary}{cor.}{corollaries}
\crefname{proposition}{prop.}{propositions}
\crefname{remark}{rem.}{remarks}
\crefname{enumi}{step}{steps}
\crefname{theorem}{thm.}{theorems}
\crefname{algorithm}{alg.}{algorithms}
\crefname{figure}{fig.}{figures}
\crefname{assumption}{assum.}{assumptions}
\crefname{appendix}{App.}{appendices}
\usepackage[ruled]{algorithm}
\usepackage[noend]{algorithmic}

\usepackage[prependcaption,textsize=tiny,textwidth=12mm]{todonotes}
\usepackage{ulem}

\usepackage{hyperref}
\usepackage{enumitem}

\def\Groebner{Gr\"{o}bner\xspace}

\def\mksum{Minkowski sum\xspace}
\def\wrt{with respect to\xspace}
\def\polyt{\Delta}
\def\R{\mathbb{R}}

\def\O{\mathcal{O}}

\def\N{\mathbb{N}}
\def\Z{\mathbb{Z}}
\def\K{\mathbb{K}}
\def\C{\mathbb{C}}
\def\Spolyt{S_{\polyt}}
\def\KS{\K[\Spolyt]}
\def\KSh{\K[\Spolyt^h]}
\def\x{\bm{x}}
\def\Kx{\K[\x]}
\def\cone{\mathcal{C}}
\def\T{(\C^*)^n}
\def\LM{\mathrm{LM}}
\def\KLau{\K[\Z^n]}
\def\NP{\mathrm{NP}}
\def\dehom{\chi}
\def\mdeg{\mathrm{mdeg}}
\def\MV{\mathrm{MV}}
\def\Ko{\mathcal{K}}
\def\H{\mathcal{H}}
\def\im{\mathrm{im}}

\def\Mac{\mathcal{M}}
\def\gMac{\widetilde{\Mac}}
\def\<{\prec}
\def\rows{\mathtt{Rows}}

\newcommand{\ideal}[1]{\langle {#1} \rangle}

\newcommand{\mon}[1]{{\text{\small $\bm{X}^{\bm{{#1}}}$}}}
\newcommand{\monhom}[2]{{\text{\small $\bm{X}^{(\bm{{#1}}, \bm{{#2}})}$}}}

\newcommand{\monhomNoBM}[2]{\bm{X}^{({{#1}}, {{#2}})}}

\def\L{\mathcal{L}}

\begin{document}
\title[\Groebner Basis over Semigroup Algebras]{
  \Groebner Basis over Semigroup Algebras: Algorithms and Applications for Sparse Polynomial Systems
}


\author{Mat\'{i}as R. Bender
  $\quad$ Jean-Charles Faug\`ere
  $\quad$ Elias Tsigaridas}

\affiliation{%
  \institution{Sorbonne Universit\'e, \textsc{CNRS}, \textsc{INRIA},
    Laboratoire d'Informatique de Paris~6, \textsc{LIP6},
    \'Equipe \textsc{PolSys}}
  \streetaddress{4 place Jussieu}
  \city{F-75005, Paris} 
  \country{France}
  }
\email{First Name.Last Name@inria.fr}
\renewcommand{\shortauthors}{
  Mat\'{i}as R. Bender,
  Jean-Charles Faug\`ere, and
  Elias Tsigaridas}

\begin{abstract}
  \Groebner bases is one the most powerful tools in algorithmic
  nonlinear algebra.
  Their computation is an intrinsically hard problem with a complexity
  at least single exponential in the number of variables. However, in
  most of the cases, the polynomial systems coming from applications
  have some kind of structure. For example, several problems in
  computer-aided design, robotics, vision, biology,
  kinematics, cryptography, and optimization involve sparse systems
  where the input polynomials have a few non-zero terms.

  Our approach to exploit sparsity is to embed the systems in a
  semigroup algebra and to compute \Groebner bases over this
  algebra. Up to now, the algorithms that follow this approach benefit
  from the sparsity only in the case where all the polynomials have
  the same sparsity structure, that is the same Newton polytope. We
  introduce the first algorithm that overcomes this restriction.  Under
  regularity assumptions, it performs no redundant computations.
  Further, we extend this algorithm to compute \Groebner basis in the
  standard algebra and solve sparse polynomials systems over the torus
  $\T$.  The complexity of the algorithm depends on the Newton
  polytopes.
\end{abstract}

%
%




\maketitle

\vspace{-8px}
\section{Introduction}


The introduction of the first algorithm to compute \Groebner bases in 1965
\cite{buchberger_bruno_2006} established them as a  central tool in
nonlinear algebra. Their applications span most of the spectrum of
mathematics and engineering  \cite{buchberger_grobner_1998}.
Computing \Groebner bases is an intrinsically hard problem.  For many
``interesting'' cases related to  applications the complexity of the
algorithms to compute them is single exponential in the number of
variables, but there are instances where the complexity is double
exponential; it is an \mbox{\textsc{EXPSPACE}} complete problem
\cite{mayr_complexity_1997}.
There are many practically efficient algorithms,
see~\cite{faugere_F5_2002,eder2014survey} and references therein,
for which,
under genericity assumptions, we can deduce precise
complexity estimates~\cite{bardet_complexity_2015}.
However, the polynomial systems coming from applications,
i.e. computer-aided design, robotics, biology, cryptography, and
optimization
e.g.,~\cite{Sturmfels-sol-bk,emiris_computer_1999,faugere:hal-01057831},
have some kind of structure.
One of the main challenges in \Groebner basis theory is to improve
the complexity and the practical performance of the related algorithms
by exploiting the structure.

We employ the structure related to the sparsity of the polynomial
systems; in other words, we focus on the non-zero terms of the input
polynomials.
In addition, we consider polynomials having different supports.
There are different approaches to benefit from sparsity,
e.g.,~\cite{sturmfels1993sparse,faugere_grobner_2011,faugere_sparse_2014,cifuentes_exploiting_2016,
  bender_towards_2018}.
We follow \cite{sturmfels1993sparse,faugere_sparse_2014} and we
consider \Groebner bases over semigroup algebras. We construct a
semigroup algebra related to the Newton polytopes of the input
polynomials and compute \Groebner bases for the ideal generated by the
original polynomials in this semigroup algebra.

We embed the system in semigroup algebras because in this place they
``behave'' in a predictable way that we can exploit algorithmically.
Semigroup algebras are related to toric varieties.  An affine toric
variety is the spectrum of a semigroup algebra
\cite[Thm.~1.1.17]{cox_toric_2011}.  Hence, the variety defined by the
polynomials over the semigroup is a subvariety of a toric variety.
This variety is different from the one defined by the polynomials over
the original polynomial algebra, but they are related and in many
applications the difference is irrelevant,
e.g.,~\cite{emiris_computer_1999}.  We refer to \cite{cox_toric_2011}
for an introduction to toric varieties and to
\cite{sturmfels_grobner_1996} for their relation with \Groebner basis.

In ISSAC'14, \citet{faugere_sparse_2014} considered sparse
\textit{unmixed systems}, that is,  polynomial systems where all
the polynomials have the same Newton polytope, and they introduced an
algorithm to compute \Groebner bases over the semigroup
algebra generated by the Newton polytope.
This algorithm is a variant of the MatrixF5 algorithm \cite{faugere_F5_2002,bardet_complexity_2015}. They compute
\Groebner basis by performing Gaussian elimination on various
Macaulay matrices \cite{lazard_grobner-bases_1983} and they avoid
computations with rows reducing to zero
using the F5 criterion \cite{faugere_F5_2002}.
The efficiency of this approach relies on an incremental construction
which, under regularity assumptions, skips all the rows reducing to
zero.
They exploit the property that, for normal Newton polytopes,
\textit{generic} unmixed systems are regular sequences over
the corresponding semigroup algebra.
Unfortunately, this property is no longer true for \textit{mixed
  systems}, that is, for systems of polynomials with different Newton
polytopes.  So, this algorithm fails to predict all rows reducing to
zero during Gaussian elimination.
Moreover, the degree bound for the
maximal degree in \cite[Lem.~5.2]{faugere_sparse_2014} misses some
assumptions to hold, see \Cref{sec:counter-example}.
We relax the regularity assumptions of \cite{faugere_sparse_2014} and
we introduce an F5-like criterion that, under regularity assumptions,
predicts all the rows reducing to zero during 
\Groebner bases computation.

In this context, we also mention our previous work
\cite{bender_towards_2018} on computing \textit{sparse \Groebner
  bases} for mixed sparse polynomial systems.  We emphasize that
besides the similarity in the titles, this work and
\cite{bender_towards_2018} are completely different approaches.
First, we compute different objects. Sparse \Groebner bases
\cite[Sec.~3]{bender_towards_2018} are not \Groebner basis for
semigroup algebras.  Moreover, we follow different computational
strategies: in \cite{bender_towards_2018} we perform the computations
polynomial by polynomial, while in this work we proceed degree by
degree.  Further,  when we use \cite{bender_towards_2018} to solve 0-dimensional
systems, there are no complexity bounds,
let alone bounds depending on the Newton polytopes, for this approach.


A direct application of \Groebner basis theory is to solve polynomial
systems. This is also an intrinsically hard problem
\cite{heintz_intrinsic_1993}. Hence, it is important to exploit the
sparsity of the input polynomials to obtain new algorithms for solving with better
complexity bounds.
The different ways of doing so include homotopy methods
e.g.,~\cite{verschelde_homotopies_1994,huber_polyhedral_1995}, chordal
elimination \cite{cifuentes_exploiting_2016}, triangular decomposition
\cite{mou_chordality_2018}, and various other techniques
\cite{giusti2001grobner,herrero2013affine,faugere_computing_2016,telen:hal-01630425,rojas1999solving,mora1982algorithm}.

Among the symbolic approaches related to toric geometry, the main tool
to solve sparse systems is the sparse resultant
\cite{gelfand_discriminants_1994}.
The resultant is a central object in elimination theory and there are
many different ways of exploiting it to solve sparse systems, see for
example \cite[Chp.~7.6]{cox_using_2005}.
Canny and Emiris \cite{canny_efficient_1993} and Sturmfels
\cite{sturmfels1994newton} showed how to compute the sparse resultant
as the determinant of a square Macaulay matrix (Sylvester-type
formula) whose rows are related to mixed subdivisions of some
polytopes.
Using this matrix, e.g.,~\cite{emiris_complexity_1996,ER-mbpss-94}, we
can solve square sparse systems. For this, we add one more polynomial
to the system and we consider the matrix of the resultant of the new
system.  Under genericity assumptions, we can recover the
multiplication maps of the quotient ring defined by original square
system over the ring of Laurent polynomials and we obtain the
solutions over $(\C \setminus \{0\})^n$.
Recently, Massri \cite{massri_solving_2016} dropped the genericity
assumptions by considering a bigger matrix.

We build on Massri's work and, under regularity
assumptions, we propose an algorithm to solve 0-dimensional square
systems with complexity related to the Minkowski sum of the
Newton polytopes.
Because we work with toric varieties, we compute  solutions  over
$(\C \setminus \{0\})^n$.
Our strategy is to reuse part of our
algorithm to compute \Groebner bases over semigroup algebras to
compute multiplication maps and, via FGLM \cite{faugere_efficient_1993}, recover a \Groebner
basis over the standard polynomial algebra. As we compute the
solutions over $(\C \setminus \{0\})^n$,
we do not recover a \Groebner
basis for the original ideal, but for its saturation \wrt the product
of all the variables.
We compute with a matrix that has the same size as the one in Emiris'
resultant approach \cite{emiris_complexity_1996}.
Our approach to solve is more general than the one in
\cite{faugere_sparse_2014} as we compute with mixed sparse
systems, and because it terminates earlier as we do not compute
\Groebner bases but multiplication maps. An overview of our
strategy is as follows:

\begin{enumerate}[leftmargin=*]
\item Let  
  $f_1,\dots,f_n \in \Kx$  be a sparse regular polynomial system
  with a finite number of
  solutions over $\T$.
\item Embed the polynomials to a multigraded semigroup algebra
  $\KSh$ related to the Newton polytopes of $f_1,\dots,f_n$ and to the
  standard n-simplex (see \Cref{deg:semigroupFromPolytope}).
\item For each variable $x_i$:
  \begin{itemize}[leftmargin=0px]
  \item Use the \Groebner basis algorithm (\Cref{alg:f5Alg}) to construct a square
    Macaulay matrix related to $(f_1,\dots,f_n,x_i)$ of size equal to
    the number of integer points in the Minkowski sum of the Newton
    polytopes of $f_1,\dots,f_n$ and the n-simplex.
  \item Split the matrix in four parts and compute a Schur
    complement, which is the multiplication map of $x_i$ in
    $\K[\bm{x}^{\pm 1}] / \ideal{f_1,\dots,f_n}$.
  \end{itemize}
\item Use the multiplication maps and FGLM to get a  \Groebner
  basis for $\ideal{f_1,\dots,f_n} : \ideal{\prod_i x_i}^\infty$ \wrt
  any monomial order.
\end{enumerate}


\vspace{-3pt}
The contributions and consequences of our work include:
\begin{itemize}[leftmargin=*]
\item \textit{We introduce the first effective algorithm to compute
    \Groebner bases over semigroup algebras associated to mixed
    polynomial systems.}
  We generalize the work of \cite{faugere_sparse_2014} to the mixed
  case from which we could provide accurate complexity estimates related to the
  Newton polytopes of the input polynomials.
  
\item \textit{We relate the solving techniques using
    Sylvester-type formulas in resultant theory with \Groebner bases computations.}
  The simplest, but not necessarily the most efficient as there are
  more compact formulas \cite{weyman_determinantal_1994}, way to
  compute the resultant is to use a Sylvester-type formula and compute
  it as the determinant of a Macaulay matrix
  \cite[Chp.~3.4]{cox_using_2005}. Using this matrix we extract
  multiplication maps and solve polynomial systems
  \cite[Chp.~3.4]{cox_using_2005}. In the standard polynomial algebra,
  such matrices are at the heart of linear algebra algorithms to
  compute \Groebner bases because they correspond to the
  biggest matrix that appears during \Groebner basis computations for
  regular 0-dimensional systems \cite{lazard_grobner-bases_1983}.
  However, such a relation was not known for the sparse case.  We
  bring out this relation and we build on it algorithmically.

\item \textit{We generalize the F5 criterion to depend on Koszul
    complexes instead of regular sequences.}
  The exactness of the Koszul complex is closely related to regular
  sequences \cite[Ch.~17]{eisenbud_commutative_2004} and,
  geometrically, to complete intersections. Roughly speaking, when we
  consider generic square systems of equations in the coordinate ring
  of a ``nice'' projective variety, the variety that the system
  defines is closely related to a complete intersection. In this case,
  the Koszul complex of the system might not be exact in general, but only in
  some ``low'' degrees. Hence, even if the system is not a regular
  sequence, by focusing on the degrees at which the strands of the Koszul
  complex are exact, we can still predict the algebraic structure of
  the system and perform efficient computations.  Using this property
  we extend the classical F5-like criteria that apply only to regular
  sequences.
  Moreover, additional information on the exactness
  of the strands of the Koszul complex and the multigraded
  Castelnuovo-Mumford regularity
  \cite{maclagan_multigraded_2004,botbol_castelnuovo_2017}
  results in better degree and complexity bounds; similarly to the 
  case of the multihomogeneous systems
  \cite[Sec.~4]{bender_towards_2018}.

\item \textit{We disrupt the classical strategy to solve 0-dimensional
    systems using \Groebner basis, by avoiding intermediate \Groebner basis computations.}
  The classical approach for  solving 0-dimensional systems using \Groebner
  bases involves the computation of a intermediate \Groebner basis
  that we use to deduce multiplication maps and,  by using FGLM, to
  obtain the lexicographical \Groebner basis of the ideal.
  If the  intermediate \Groebner basis is computed \wrt a graded
  reverse lexicographical order and the input system ``behaves well''
  when we homogenize it, this strategy is some sense optimal because it
  is related to the Castelnuovo-Mumford regularity of the homogenized
  ideal \cite[Cor.~3]{chardin_bounds_2003}.

  However, over semigroup algebras, it might not be always possible to
  relate the complexity of the intermediate \Groebner basis
  computation to the Castelnuovo-Mumford regularity of the ideal;
  this is so 
  because we can not define monomial orders that behave like a graded
  reverse lexicographical, see \cite[Ex.~2.3]{bender_towards_2018}.
  We overcome this obstacle by truncating the computation of the
  intermediate \Groebner basis in such a way that the complexity is
  given by Castelnuovo-Mumford regularity of the ideal.
\end{itemize}

\vspace{-15pt}
\section{Preliminaries}
\label{sec:preliminaries}

Let $\K \subset \C$ be a field of characteristic $0$,
$\x := (x_1,\dots,x_n)$,
and $\Kx := \K[x_1,\dots,x_n]$.
We consider $\bm{0} := (0,\dots,0)$ and
$\bm{1} := (1,\dots,1)$. For each $r \in \N$, let
$\bm{e_1},\dots,\bm{e_r}$ be the canonical basis of $\R^{r}$.
Given  $\bm{d_1},\bm{d_2} \in \N^r$, we say
$\bm{d_1} \geq \bm{d_2}$ when $\bm{d_1} - \bm{d_2} \in \N^r$.
We use $[r] = \{1, \dots, r \}$.
We denote by $\ideal{f_1, \dots, f_m}$ the ideal generated by
$f_1, \dots, f_m$.

\subsection{Semigroup algebras}

\begin{definition}[Affine semigroup and semigroup algebra]
Following \cite{miller_combinatorial_2005}, an affine semigroup $S$ is
a finitely-generated additive subsemigroup of $\Z^n$, for some
$n \in \N$, such that it contains ${\bm{0}} \in \Z^n$.
An affine semigroup $S$ is pointed if it does not contain non-zero
invertible elements, that is for all
$\bm{\alpha}, \bm{\beta} \in S \setminus \{ \bm{0}\}$, $ \bm{\alpha}+ \bm{\beta} \neq \bm{0}$
\cite[Def~7.8]{miller_combinatorial_2005}.
The semigroup algebra $\K[S]$ is the $\K$-algebra generated by the
monomials $\{\mon{\bm{\alpha}} : \bm{\alpha} \in S\}$ such that
$\mon{\bm{\alpha}} \cdot \mon{\bm{\beta}} = \mon{\bm{\alpha} + \bm{\beta}}$.
\end{definition}


\begin{definition}
  [Convex set and convex hull] A set $\Delta \subset \R^n$ is convex
  if every line segment connecting two elements of $\Delta$ also lies
  in $\Delta$; that is, for every $\bm{\alpha},\bm{\beta} \in \Delta$
  and $0 \leq \lambda \leq 1$ it holds
  $\lambda \bm{\alpha} + (1 - \lambda) \bm{\beta} \in \Delta$. The
  convex hull of $\Delta$ is the unique minimal, \wrt inclusion,
  convex set that contains $\Delta$.
\end{definition}

\begin{definition}[Pointed rational polyhedral cones]
  A cone $\cone$ is a convex subset of $\R^n$ such that $\bm{0}\in \cone$
and for every $\bm\alpha \in \cone$ and $\lambda > 0$,
$\lambda \, \bm\alpha \in \cone$. The \emph{dimension of a cone} is the dimension of
the vector space spanned by the cone. 
A cone is pointed if does not contain any line; that is, if
$\bm 0 \neq \bm\alpha \in \cone$, then $- \bm\alpha \not\in \cone$. A ray is a
pointed cone of dimension one. A ray is
\textit{rational} if it contains a non-zero point of $\Z^n$. A
\textit{rational polyhedral cone} is the convex hull of a finite set
of rational rays. For  a set of points $\Delta \subset \R^n$, let
$\cone_{\Delta}$ be the cone generated by the elements in $\Delta$. If
$\Delta$ is (the convex hull of) a finite set of integer points, then
$\cone_{\Delta}$ is a rational polyhedral cone.
\end{definition}

A rational polyhedral cone $\cone$ defines the affine semigroup
$\cone \cap \Z^n$, which is pointed if and only if the cone is
pointed.

\begin{definition}[Integer polytopes and Minkowski sum]
  A integer polytope $\polyt \subset \R^n$ is the convex hull of a
  finite set of (integer) points in $\Z^n$.
  The Minkowski sum of two integer polytopes $\polyt_1$ and $\polyt_2$ is
  $\polyt_1 + \polyt_2 = \{\bm\alpha + \bm\beta : \bm\alpha \in \polyt_1, \bm\beta
  \in \polyt_2\}$. For each polytope $\polyt$ and $k \in \N$, we denote by 
  $k \!\cdot\! \polyt$ the  Minkowski sum of $k$ copies of $\polyt$.
\end{definition}

\begin{definition}[Laurent polynomials and Newton polytopes]
  A Laurent polynomial is a finite $\K$-linear combination of monomials
  $\mon{\bm\alpha}$, where  $\bm\alpha \in \Z^n$. The Laurent
  polynomials form a ring, $\KLau$, that corresponds to the semigroup algebra
  of $\Z^n$.
  For a Laurent polynomial
  $f = \sum_{\bm\alpha \in \Z^n} c_{\bm\alpha} \bm{x}^{\bm\alpha}$, its
  Newton polytope is the integer polytope generated by the set of
  the exponents $\bm{\alpha}$ of the non-zero coefficients of $f$; that is,
  $\NP(f) := \text{Convex Hull}(\{\bm\alpha \in \Z^n, c_{\bm{\alpha}} \neq 0\})$.
\end{definition}

Instead of working over $\KLau$, we embed $f$ in a subalgebra
related to its Newton polytope, given by
$\K[\cone_{\NP(f)} \cap \Z^n]$. In this way we  exploit the
sparsity of the (polynomials of the) system.

\begin{definition}[Semigroup algebra of polytopes]
  \label{deg:semigroupFromPolytope}
  
  We consider $r$ integer polytopes
  $\polyt_1,\dots,\polyt_r \subset \R^n$ such that their Minkowski
  sum, $\polyt := \sum_{i=1}^{r} \polyt_i$, has dimension $n$ and  $\bm{0}$
  is its vertex; in particular, $\bm{0}$ as a vertex of every Newton
  polytope $\polyt_i$.
  We also consider the polytope
  $\bar{\polyt} := \sum (\polyt_i \times \{\bm{e_i}\})$, which is the
  Cayley embedding of $\polyt_1,\dots,\polyt_r$.

  In what follows, we  work with the semigroup algebras
  $\KS := \K[\cone_\polyt \cap \Z^n]$ and
  $\KSh := \K[\cone_{\bar{\polyt}} \cap \Z^{n + r}]$.
  We will write the monomials in $\KSh$ as $\monhom{\alpha}{d}$, where
  $\bm{\alpha} \in (\cone_\polyt \cap \Z^n)$ and $\bm{d} \in \N^r$.
\end{definition}

  The algebra  $\KSh$ is $\N^{r}$-multigraded as follows: for every
  $\bm{d} = (d_1,\dots,d_r) \in \N^{r}$, $\K[S_{\polyt}^h]_{\bm{d}}$
  is the $\K$-vector space spanned by the monomials
  $\{\monhom{\alpha}{d} : \bm{\alpha} \in (\sum d_i \cdot \polyt_i)
  \cap \Z^n \}$.
  Then, $F \in \KSh_{\bm{d}}$
  is homogeneous and has multidegree
  $\bm{d}$, which we denote by $\mdeg(F)$.

We can think  $\KS$ as the ``dehomogenization'' of $\KSh$.

\begin{definition}[Dehomogenization morphism]
  \label{def:dehomMorph}
  The dehomogenization morphism from $\KSh$ to $\KS$ is the
  surjective ring homomorphism $\dehom : \KSh \rightarrow \KS$ that
  maps the monomials $\monhom{\alpha}{\bm{d}} \in \KSh$ to
  $\dehom(\monhom{\alpha}{\bm{d}}) := \mon{\alpha} \in \KS$.

  If $\L$ is a set of homogeneous polynomials in $\KSh$,
  then we consider $\dehom(\L) = \{\dehom(G) : G \in \L\}$.
\end{definition}

\begin{observation}
  \label{obs:embedInHomogeneous}
  As $\bm{0}$ is a vertex of $\polyt$,
  there is a monomial $\monhom{0}{e_i} \in \KSh$, for every $i \in [r]$.
  Hence, given a
  finite set of monomials
  $\mon{\alpha_1},\dots,\mon{\alpha_k} \in \KS$, we can find a
  multidegree $\bm{d} \in \N^r$ such that
  $\monhom{\alpha_1}{d},\dots,\monhom{\alpha_k}{d} \in \KS_{\bm{d}}$.

  Given a system of polynomials $f_1,\dots,f_m \in \KS$, we can find a
  multidegree $\bm{d} \in \N^r$ and homogeneous polynomials
  $F_1,\dots,F_m \in \KSh_{\bm{d}}$ so that it holds
  $\dehom(F_i) = f_i$, for every $i \in [m]$.

  Moreover, given homogeneous polynomials $F_1,\dots,F_m \in \KSh$ and
  an affine polynomial $g \in \ideal{\dehom(F_1),\dots,\dehom(F_m)}$,
  there is an homogeneous polynomial $G \in \ideal{F_1,\dots,F_m}$
  such that $\dehom(G) = g$. 
\end{observation}

\begin{observation}
  \label{thm:dehomSameDegIsInjective}
  If we fix a multidegree $\bm{d} \in \N^r$, then the map $\dehom$ restricted to
  $\KSh_{\bm{d}}$ is injective.
\end{observation}



\subsection{\Groebner bases}

We recall some definitions related to Groebner basis over semigroup
algebras from \cite{faugere_sparse_2014}.  Let $S$ be a pointed affine
semigroup.

\begin{definition}
  [Monomial order]

  Given a pointed semigroup algebra $\K[S]$, a monomial order for
  $\K[S]$, say $<$, is a total order for the monomials in $\K[S]$ such that:
  \begin{itemize}
  \item For any $\bm{\alpha} \in S \setminus \{\bm{0}\}$, it holds $\mon{\bm{0}} < \mon{\bm{\alpha}}$.
  \item For every $\bm{\alpha},\bm{\beta},\bm{\gamma} \in S$, if $\mon{\bm{\alpha}} < \mon{\bm{\beta}}$ then
    $\mon{\bm{\alpha} + \bm{\gamma}} < \mon{\bm{\beta} + \bm{\gamma}}$.
  \end{itemize}
\end{definition}

\begin{observation}
  Monomial orders always exist for pointed affine semigroups.
  To construct them, first we  embed any pointed affine semigroup of dimension $n$ in
  a pointed rational cone $\cone \subset \R^n$. Then, we choose $n$
  linearly independent forms $l_1,\dots,l_n$ from the dual cone of $\cone$,
  which is 
  $\{ l : \R^n \rightarrow \R \mid
  \forall \bm{\alpha} \in \C, l( \bm{\alpha}) \geq 0\}$.
  We  define the monomial order so that 
  $\mon{\bm{\alpha}} < \mon{\bm{\beta}}$ if and only if there is a $k \leq n$ such
  that for all $i < k$  it  holds $ l_i(\bm{\alpha}) = l_i(\bm{\beta})$ and
  $l_k(\bm{\alpha}) < l_k(\bm{\beta})$.
\end{observation}

\begin{definition}[Leading monomial]
  Given a monomial order $<$ for a pointed affine semigroup algebra
  $\K[S]$ and a polynomial $f \in \K[S]$, its leading monomial,
  $\LM_<(f)$ is the biggest monomial of $f$ \wrt the monomial order
  $<$.
\end{definition}

The exponent of the leading monomial of $f$ always corresponds to a
vertex of $\NP(f)$.

\begin{definition}[\Groebner basis]
  Let $\K[S]$ be a pointed affine semigroup algebra and consider a
  monomial order $<$ for $\K[S]$.  For an ideal $I \subset \K[S]$,
  a set $G \subset I$ is a \Groebner basis of $I$ if
  $\{LM_<(g) : g \in G\}$ generates the same ideal as
  $\{LM_<(f) : f \in I\}$.

  In other words, if for every $f \in I$,
  there is $g \in G$ and $\mon{\alpha} \in \K[S]$ such that
  $\LM_<(f) = \mon{\alpha} \, \LM_<(g)$.
\end{definition}

As $S$ is finitely generated, the algebra $\K[S]$ is a Noetherian ring
\cite[Thm.~7.7]{gilmer_commutative_1984}.  Hence, for any monomial
order and any ideal, there is always a finite \Groebner basis.


We will consider monomial orders for $\KSh$ that
we can relate to monomial orders in $\KS$ and $\K[\N^r]$.

\begin{definition}
  [Multigraded monomial order]
  \label{def:mgraded-order}
  
  We say that a monomial order $<$ for
  $\KSh$ is multigraded, if there are monomial orders
  $<_\polyt$ for $\KS$ and $<_h$ for $\K[\N^r]$ such that,
  for every
  $\monhom{\alpha_1}{d_1},\monhom{\alpha_2}{d_2} \in \KSh$,
  it holds 
  \begin{equation} \label{eq:graduateOrder}
    \monhom{\alpha_1}{d_1} < \monhom{\alpha_2}{d_2}
    \iff
    \begin{cases}
      \mon{d_1} <_{h} \mon{d_2}  \text{ or} \\
      \bm{d_1} = \bm{d_2} \, \text{ and } \, \mon{\alpha_1} <_{\polyt} \mon{\alpha_2}
    \end{cases} .
  \end{equation}
\end{definition}

Multigraded monomial orders are ``compatible'' with the
dehomogenization morphism (\Cref{def:dehomMorph}).

\begin{remark}
  \label{rem:associated-order}
  In what follows, given a multigraded
  monomial order $<$ for $\KSh$, we also use the same symbol, that is  $<$,
  for the associated monomial order of $\KS$.
\end{remark}

\begin{lemma}
  \label{thm:monOrderCommuteDehom}
  Consider a polynomial $f \in \KS$. Let $<$ be a multigraded monomial
  order. For any multidegree $\bm{d}$ and any homogeneous
  $F \in \KSh_{\bm{d}}$ such that $\dehom(F) = f$, it holds
  $
  \LM_{<}(f) = \dehom(\LM_{<}(F)).
  $
\end{lemma}

\subsection{Regularity and solutions at infinity}
\label{sec:reg-inf-sol}

The Bernstein-Kushnirenko-Khovanskii (BKK) theorem bounds the (finite)
number of solutions of a square system of sparse Laurent
polynomials over the torus $\T$, where $\C^* := \C \setminus \{0\}$.

\begin{definition}
  [Mixed volume]

  Let $\polyt_1,\dots,\polyt_n \in \R^n$ be integer polytopes.
  Their
  mixed volume, $\MV(\polyt_1,\dots,\polyt_n)$, is the alternating sum
  of the number of integer points of the polytopes obtained
  by all possible  Minkowski sums, that is 
  {\footnotesize
    \begin{multline}
      \label{eq:mixedVolume}
      \MV(\polyt_1,\dots,\polyt_n) =  \\
      (-1)^n
      +
      \sum_{k = 1}^n (-1)^{n-k}
      \Big(
      \sum_{\substack{I \subset \{1,\dots,n\} \\ \#I = k}}
      \!\!\!\!
      \# \left(
        (\polyt_{I_1} + \dots + \polyt_{I_k})\cap \Z^n
      \right)
      \Big).
    \end{multline}
  }
\end{definition}

\begin{theorem}[{BKK bound \cite[Thm~7.5.4]{cox_using_2005}}]

  Let $f_1,\dots,f_n$ be a system of polynomials with Newton polytopes
  $\polyt_1,\dots,\polyt_n$ having  a finite number of solutions over
  $\T$.
  The mixed volume
  $\MV(\polyt_1,\dots,\polyt_n)$ upper bounds the number of solutions
  of the system over the torus $\T$.
  If the non-zero coefficients of the polynomials
  are generic, then the bound is tight.
\end{theorem}

Toric varieties relate semigroup algebras with the torus $\T$.  A
toric variety is an irreducible variety $X$ that contains $\T$ as an
open subset such that the action of $\T$ on itself extends to an
algebraic action of $\T$ on $X$ \cite[Def.~3.1.1]{cox_toric_2011}.
Semigroup algebras correspond to the coordinate rings of the affine
pieces of $X$.

Given an integer polytope $\polyt$, we can define a projective
complete normal irreducible toric variety $X$ associated to it
\cite[Sec.~2.3]{cox_toric_2011}.
Likewise, given a polynomial system $(f_1,\dots,f_m)$, we can define
a projective toric variety $X$ associated to the \mksum of
their Newton polytopes.
We can homogenize these polynomials in a way that they belong to the
total coordinate ring of $X$ 
\cite[Sec.~5.4]{cox_toric_2011}. This homogenization is related to the
facets of the polytopes.

To be more precise, given an integer polytope $\polyt \subset \R^n$, we say
that an integer polytope $\polyt_1$ is a \textit{$\N$-Minkowski summand} of
$\polyt$ if there is a $k \in \N$ and another polytope $\polyt_2$ such
that $\polyt_1 + \polyt_2 = k \!\cdot\! \polyt$
\cite[Def.~6.2.11]{cox_toric_2011}.
Every $\N$-Minkowski summand
$\polyt_1$ of $\polyt$ defines a torus-invariant basepoint free
Cartier divisor $D$ of the projective toric variety $X$ associated to
$\polyt$ \cite[Cor.~6.2.15]{cox_toric_2011}. This divisor defines an
invertible sheaf $\O_X(D)$ whose global sections form the vector space
of polynomials in $\K[\Z^n]$ whose Newton polytopes are contained in
$\polyt_1$ \cite[Lem.~1]{massri_solving_2016}.  Therefore, to
homogenize $f_1,\dots,f_m$ over $X$ we need to choose polytopes
$\polyt_1,\dots,\polyt_m$ such that all of them are $\N$-Minkowski
summands of $\polyt$ associated to $X$ and
$\NP(f_i) \subset \polyt_i$.
Hence, for any homogeneous $F \in \KSh_{\bm{d}}$, we can homogenize
$\dehom(F)$ with respect to the $\N$-Minkowski summand
$\sum_i d_i \polyt_i$ of $\Delta$.

We alert the reader that homogeneity in $\KSh_{\bm{d}}$ is different from
homogeneity in the total coordinate ring of $X$,
see  \cite[Sec.~5.4]{cox_toric_2011}
but they are related through the degree $\bm{d}$.

\begin{definition}
  [Solutions at infinity]
  \label{def:sol-inf}
  
  Let $(f_1,\dots,f_m)$ be a system of polynomials.  Let $X$ be the
  projective toric variety associated to a polytope $\polyt$ such that
  the Newton polytope of $f_i$ is a $\N$-Minkowski summand
  of $\polyt$, for all $i$.
  We say that the system has no solutions at infinity \wrt $X$ if the
  homogenized system \wrt their Newton polytopes has no solutions
  over $X \setminus \T$.
\end{definition}




\begin{proposition}
  [{\cite[Thm.~3]{massri_solving_2016}}]
  \label{thm:noSolsAtInfBKKtight}
  
  Consider a system $(f_1,\dots,f_n)$ having finite number of
  solutions over $\T$.  Let $X$ be the projective toric variety
  associated to the corresponding
  Newton polytopes. Then, the number of solutions of the homogenized
  system over $X$, counting multiplicities, is exactly the BKK
  bound.
  When the original system has no solutions at
  infinity, then the BKK is tight over $\T \subset X$.
\end{proposition}

\begin{definition} 
  [Koszul complex, {\cite[Sec.~17.2]{eisenbud_commutative_2004}}]
  \label{def:KoszulComplex}
  For a sequence of homogeneous  $F_1,\dots,F_k \in \KSh$
  of multidegrees $\bm{d}_1,\dots,\bm{d}_k$ and a multidegree
  $\bm{d} \in \N^r$, we denote by $\Ko(F_1,\dots,F_k)_{\bm{d}}$ the
  strand of the Koszul complex of $F_1,\dots,F_k$ of multidegree
  $\bm{d}$, that is,
  
  \vspace{\abovedisplayshortskip}
    \hfil
  $
  \Ko(F_1,\dots,F_k)_{\bm{d}} :
  0 \rightarrow (\Ko_{k})_{\bm{d}} \xrightarrow{\delta_k} \dots 
  \xrightarrow{\delta_1} (\Ko_{0})_{\bm{d}} \rightarrow 0,
  $

  \noindent
  where, for $1 \leq t \leq k$, we have 
  $$
  (\Ko_{t})_{\bm{d}} :=
  \bigoplus_{\substack{I \subset \{1,\dots,k\} \\ \#I = t}}
  \KSh_{( \bm{d} - \sum\limits_{i \in I} \bm d_{i} )}
  \otimes
  (e_{I_1} \wedge \dots \wedge {e_{I_t}}).
  $$
  The maps (differentials) act as follows:
  {\small
    \begin{multline}
      \delta_t \Big( \sum_{\substack{I \subset \{1,\dots,k\} \\ \#I = t}}
      g_I \otimes (e_{I_1} \wedge \dots \wedge {e_{I_t}}) \Big) = \\
      \sum_{\substack{I \subset \{1,\dots,k\} \\ \#I = t}} \sum_{i = 1}^t
      (-1)^{i-1} F_{I_i} \, g_I \otimes (e_{I_1} \wedge \dots \wedge
      \widehat{e_{I_i}} \wedge \dots \wedge {e_{I_t}}).
    \end{multline}
  }

  The expression
  $(e_{I_1} \wedge \dots \wedge \widehat{e_{I_i}} \wedge \dots \wedge
  {e_{I_t}})$ denotes that we skip the term ${e_{I_i}}$ from the
  wedge product.
  We denote by $\H_t(F_1,\dots,F_k)_{\bm{d}}$ the $t$-th Koszul
  homology of $\Ko(F_1,\dots,F_k)_{\bm{d}}$, that is
  $\H_t(F_1,\dots,F_k)_{\bm{d}} := (\ker(\delta_{t}) / \im(\delta_{t+1}))_{\bm{d}}.$
\end{definition}

\noindent
The $0$-th Koszul homology is
$
\H_0(F_1,\dots,F_k) \cong 
(\KSh / \ideal{F_1,\dots,F_k}). $

\begin{definition}
  [Koszul and sparse regularity]
  \label{def:Koszul-reg}
  A sequence \linebreak $F_1,\dots,F_k \in \KSh$ is Koszul regular if for every
  $\bm{d} \in \N^r$ coordinate-wise greater than or equal to
  $\bm{D_k} := \sum_{i = 1}^k \bm d_i$, that is,
  $\bm d \geq \bm{D_k}$,
  and for every $t > 0$, the $t$-th Koszul
  homology vanishes at degree $\bm{d}$,
  that is $\H_t(F_1,\dots,F_k)_{\bm{d}} = 0$.
  We say that the sequence is (sparse) regular if $F_1,\dots,F_j$ is
  Koszul regular, for every $j \leq k$.
\end{definition}

\begin{observation}
  Note that Koszul regularity does not depend on the order of the
  polynomials, as (sparse) regularity does.
\end{observation}

Kushnirenko's proof of the BKK bound
\cite[Thm.~2]{kushnirenko_newton_1976} follows from Koszul
regularity.



\section{The algorithm}
\label{sec:algo}

To compute \Groebner basis over $\KS$ we work over $\KSh$. We follow
the classical approach of Lazard \cite{lazard_grobner-bases_1983}
adapted to the semigroup case, see also \cite{faugere_sparse_2014}; we
``linearize'' the problem by reducing the \Groebner basis computation
to a linear algebra problem.

\begin{lemma}
  \label{thm:mdegBigEnoughThenGB}
  Consider $F_1,\dots,F_m \in \KSh$ and a multigraded monomial order
  $<$ for $\KS$ (\Cref{def:mgraded-order}).
  There is a
  multidegree $\bm{d}$ and  homogeneous
  $\{G_1,\dots,G_t\} \subset \ideal{F_1, \dots, F_m} \cap \KSh_{\bm{d}}$
  such that $\{\dehom(G_1),\dots,\dehom(G_t)\}$ is a \Groebner
  basis of the ideal $\ideal{\dehom(F_1),\dots,\dehom(F_m)}$ \wrt the
  associated monomial order $<$ (\Cref{rem:associated-order}).
\end{lemma}
\begin{proof}
  Let $g_1,\dots,g_t \in \KS$ be a \Groebner basis for the ideal
  $\ideal{\dehom(F_1),\dots,\dehom(F_m)}$ \wrt $<$.  By
  \Cref{obs:embedInHomogeneous}, there are polynomials
  $\bar{G}_1,\dots,\bar{G}_t \in \ideal{F_1,\dots,F_m}$ such that
  $\dehom(\bar{G}_i) = g_i$, for $i \in [t]$.
  Consider $\bm{d} \in \N^r$ such that
  $\bm{d} \geq \mdeg(\bar{G}_i)$, for $i \in [t]$.
  It suffices to consider  
  $G_i = \monhom{0}{\bm{d} \text{\unboldmath{$- \mdeg(\bar{G}_i)$}}} \, \bar{G}_i \in
  \KSh_{\bm{d}}$, for $i \in [t]$.
\end{proof}

When we know a multidegree $\bm{d}$ that satisfies
\Cref{thm:mdegBigEnoughThenGB}, we can compute the \Groebner basis
over $\KS$ 
using linear algebra. For this task we need to introduce the Macaulay matrix.

\begin{definition}[Macaulay matrix] \label{thm:defMacaulayMatrix}
  A Macaulay matrix $\Mac$ of degree $\bm{d} \in \N^r$ \wrt a
  monomial order $<$ is a matrix whose columns are indexed by all
  monomials $\monhom{\alpha}{d} \in \KSh_{\bm{d}}$ and the rows by
  polynomials in $\KSh_{\bm{d}}$.
  The indices of the columns are sorted in decreasing
  order \wrt $<$.
  The element of $\Mac$ whose row corresponds to a polynomial $F$
  and whose column corresponds to a monomial $\monhom{\alpha}{d}$
  is the coefficient of the monomial $\monhom{\alpha}{d}$ of $F$.
  %
  Let $\mathtt{Rows}(\mathcal{M})$ be the set of
  \emph{non-zero} polynomials that index the rows of $\mathcal{M}$
  and $\LM_{<}(\mathtt{Rows}(\mathcal{M}))$ be the set of leading
  monomials of these polynomials.
\end{definition}

\begin{remark} \label{rem:leadMonInMacMatrix} As the columns
  of the Macaulay matrices are sorted in decreasing order
  \wrt a monomial order, the leading monomial of a polynomial
  associated to a row corresponds to the index of the column of the
  first non-zero element in this row.
\end{remark}

\begin{definition}
  \label{def:GMacaulayMatrix}
  Given a Macaulay matrix $\Mac$,
  let $\gMac$ be a new Macaulay matrix corresponding to the row
  echelon form of $\Mac$.  We can compute $\gMac$ by applying
  Gaussian elimination to $\Mac$.
\end{definition}

\begin{remark}
  \label{rem:Row-op-same-ideal}
  When we perform row operations (excluding multiplication by 0) to
  a Macaulay matrix, we do not change the ideal spanned by the
  polynomials corresponding to its rows.
\end{remark}

We use Macaulay matrices to compute a basis for the vector space
$\ideal{F_1,\dots,F_k}_{\bm{d}} := \ideal{F_1,\dots,F_k} \cap
\KSh_{\bm{d}}$ by Gaussian elimination.

   \begin{lemma}
     \label{thm:leadingMonOfMacMatrix}
     Consider homogeneous polynomials
     $F_1,\dots,F_k \in \KSh$ of multidegrees
     $\bm{d_1},\dots,\bm{d_k}$ and a multigraded monomial order $<$.
     Let $\Mac_{\bm{d}}^k$ be the Macaulay matrix whose rows
     correspond to the polynomials that we obtain
     by considering the 
     product of every monomial of multidegree ${\bm{d} - \bm{d_i}}$
     and every polynomial $F_i$;
     that is 
     \begin{equation}
       \label{eq:allMonvsPol}
       \!\!\!\!\!
       \rows(\Mac_{\bm{d}}^k) = 
       \Big\{
       \monhom{\alpha}{d - d_i} F_i : i \in [k], \monhom{\alpha}{d - d_i} \!\in\! \KSh_{\bm{d} - \bm{d_i}}
       \Big\}.
     \end{equation}
     Let $\gMac_{\bm{d}}^k$ be the row echelon form of the Macaulay matrix
     $\Mac_{\bm{d}}^k$ (\Cref{def:GMacaulayMatrix}).

     Then, the set
     of the leading monomials of the polynomials in
     $\rows(\gMac_{\bm{d}}^k)$ \wrt $<$
     is the set 
     of all the leading monomials of the ideal $\ideal{F_1,\dots,F_k}$ at
     degree $\bm{d}$.
   \end{lemma}
   \begin{proof}
     We prove that $\LM_<(\rows(\gMac_{\bm{d}}^k)) = \LM_<(\ideal{F_1, \dots, F_k}_{\bm{d}})$.
     First, we show that
     $\LM_<(\rows(\gMac_{\bm{d}}^k)) \supseteq \LM_<(\ideal{F_1,
       \dots, F_k}_{\bm{d}})$.
     Let $G$ be a polynomial in the vector space of polynomials of
     degree $\bm{d}$ in $\ideal{F_1,\dots,F_k}$. This vector space,
     $\ideal{F_1, \dots, F_k}_{\bm{d}}$,  is isomorphic to the row
     space of $\Mac_{\bm{d}}^k$, which, in turn, is the same as the
     row space of $\gMac_{\bm{d}}^k$, by \Cref{rem:Row-op-same-ideal}.
     Hence, there is a vector $v$ in the row space of
     $\gMac_{\bm{d}}^k$ that corresponds to $G$.
     Let $s$ be the index of the first non-zero element of $v$.  As
     $\gMac_{\bm{d}}^k$ is in row echelon form and $v$ belongs to its
     row space, there is a row of $\gMac_{\bm{d}}^k$ such that its
     first non-zero element is also at the $s$-th position.
     Let $F$ be the  polynomial that corresponds to this row.
     Finally, the leading monomials of the
     polynomials $F$ and $G$ are the same, that is
     $LM_<(G) = LM_<(F)$, by \Cref{rem:leadMonInMacMatrix}.

     The other direction is straightforward.
   \end{proof}

   \begin{theorem} \label{thm:correctnessGB}
     Consider the ideal generated by homogeneous polynomials
     $F_1,\dots,F_k \in \KSh$ of multidegrees
     $\bm{d_1},\dots,\bm{d_k}$. Consider a multigraded monomial order $<$ and a
     multidegree $\bm{d} \in \N^r$ that satisfy 
     \Cref{thm:mdegBigEnoughThenGB}.
     Let $\Mac_{\bm{d}}^k$ and $\gMac_{\bm{d}}^k$ be the Macaulay
     matrices of \Cref{thm:leadingMonOfMacMatrix}.

     Then, the set
     $\dehom(\mathtt{Rows}(\gMac_{\bm{d}}^k))$, see~\Cref{def:dehomMorph},
     contains a \Groebner
     basis of the ideal
     $\ideal{\dehom(F_1),\dots,\dehom(F_k)} \subset \KS$ \wrt $<$.
   \end{theorem}

   \begin{proof}
     Let $R := \mathtt{Rows}(\gMac_{\bm{d}}^k)$ be the set of
     polynomials indexing the rows of $\gMac_{\bm{d}}^k$.
     By \Cref{thm:leadingMonOfMacMatrix},
     for every
     $G \in \ideal{F_1,\dots,F_k}_{\bm{d}}$ there is a $F \in R$ such
     that $LM_<(G) = LM_<(F)$.
     As $<$ is a multigraded order, it holds
     $LM_<(\dehom(G)) = LM_<(\dehom(F))$
     (\Cref{thm:monOrderCommuteDehom}).
     As $\bm{d}$ satisfies \Cref{thm:mdegBigEnoughThenGB}
     for every 
     $h \in \ideal{\dehom(F_1),\dots,\dehom(F_k)}$ there is
     $G \in \ideal{F_1,\dots,F_k}_{\bm{d}}$ such that
     $LM_<(\dehom(G))$ divides $LM_<(h)$.
     Hence, there is an $F \in R$
     such that $LM_<(\dehom(F))$ divides $LM_<(h)$.
     Therefore, $R$ is a
     \Groebner basis for $\ideal{\dehom(F_1),\dots,\dehom(F_k)}$.
   \end{proof}

   Theorem~\ref{thm:correctnessGB} leads to an algorithm for computing
   \Groebner bases through a Macaulay matrix and Gaussian
   elimination.

    
     


   \setlength{\textfloatsep}{0pt}
   {\footnotesize
     \begin{algorithm}[t]
    \caption{\texttt{ComputeGB}  
    }
     \begin{algorithmic}[1]
     \label{alg:lazardAlg}
     \REQUIRE $f_1,\dots,f_k \in \KS$, a monomial order $<$.

     \ENSURE \Groebner basis for $\ideal{f_1,\dots,f_k}$ \wrt $<$.

     \FORALL{$f_i$}
     
     \STATE Choose $F_i \in \KSh_{\bm{d_i}}$
     of multidegree $\bm{d_i}$ such that
     $\dehom(F_i) = f_i$.
     
     \ENDFOR

     \STATE Pick a big enough $\bm{d} \in \N^r$ that satisfies
     \Cref{thm:mdegBigEnoughThenGB}.

     \STATE $\Mac_{\bm{d}}^k \leftarrow$ \parbox[t]{180px}{Macaulay matrix of
     multidegree $\bm{d}$ \wrt a multigraded monomial order associated to $<$.}
     
     \FORALL{$F_i$}
     \FORALL{$\monhom{\alpha}{d - d_i} \!\in\! \KSh_{\bm{d} - \bm{d_i}}$}
       \STATE Add the polynomial  $\monhom{\alpha}{d - d_i} F_i$ as row to $\Mac_{\bm{d}}^k$.
     \ENDFOR
     \ENDFOR

     \STATE $\gMac^{k}_{\bm{d}} \leftarrow$ \texttt{GaussianElimination}($\Mac^k_{\bm{d}}$)      
             
     \RETURN $\dehom(\rows(\gMac^k_{\bm{d}}))$
     \end{algorithmic}
   \end{algorithm}
   }

\subsection{Exploiting the structure of Macaulay matrices (Koszul F5 criterion)}

If we consider all the polynomials of the set in
\Cref{eq:allMonvsPol}, then many of them are linearly dependent.
Hence, when we construct the Macaulay matrix of
\Cref{thm:correctnessGB} and perform Gaussian elimination, many of the
rows reduce to zero; this forces \Cref{alg:lazardAlg} to perform
unnecessary computations.  We will extend to F5 criterion
\cite{faugere_F5_2002} in our setting to avoid redundant computations.

\begin{theorem}[Koszul F5 criterion]
  \label{thm:f5Crit}
  Consider homogeneous polynomials $F_1,\dots,F_k \in \KSh$ of
  multidegrees $\bm{d_1},\dots,\bm{d_k}$ and a multidegree
  $\bm{d} \in \N^r$ such that $\bm{d} \geq \bm{d_k}$, that is
  coordinate-wise greater than or equal to $\bm{d_k}$.
  Let $\Mac_{\bm{d}}^{k-1}$ and $\Mac_{\bm{d} - \bm{d_k}}^{k-1}$ be
  the Macaulay matrices of degrees $\bm{d}$ and
  ${\bm{d} - \bm{d_{k}}}$, respectively, of the polynomials
  $F_1,\dots,F_{k-1}$ as in \Cref{thm:correctnessGB}, and let
  $\gMac_{\bm{d}}^{k-1}$ and $\gMac_{\bm{d} - \bm{d_{k}}}^{k-1}$ be
  their row echelon forms.
  
  For every
  $\monhom{\alpha}{\bm{d} - \bm{d_{k}}} \in
  \LM_{<}(\rows(\gMac_{\bm{d} - \bm{d_{k}}}^{k-1}))$, the
  polynomial $\monhom{\alpha}{\bm{d} - \bm{d_{k}}} F_k$ is
  a linear combination of the polynomials
  \begin{align*}
     \rows(\gMac_{\bm{d}}^{k-1})
    \cup
       \left\{
    \monhom{\beta}{d - d_k} \, F_k :
    \begin{array}{l}
      \monhom{\beta}{d - d_k} \in \KSh_{\bm{d} - \bm{d_k}} \text{ and } \\
      \monhom{\beta}{d - d_k} < \monhom{\alpha}{d - d_k}
    \end{array}
    \right\}.
  \end{align*}
\end{theorem}

\begin{proof}
  If
  $\monhom{\alpha}{\bm{d} - \bm{d_{k}}} \in
  \LM_{<}(\rows(\gMac_{\bm{d} - \bm{d_{k}}}^{k-1}))$, then there is
  $G \in \KSh_{\bm{d} - \bm{d_{k}}}$ such that
  $\monhom{\alpha}{\bm{d} - \bm{d_{k}}} + G \in
  \ideal{F_1,\dots,F_{k-1}}_{\bm{d} - \bm{d_{k}}}$ and
  $\monhom{\alpha}{\bm{d} - \bm{d_{k}}} > \LM_{<}(G)$. So, there are
  homogeneous $H_1,\dots,H_{k-1} \in \KSh$ such that
  $\monhom{\alpha}{\bm{d} - \bm{d_{k}}} + G = \sum_i H_i F_i$.  The
  proof follows by noticing that
  $\monhom{\alpha}{\bm{d} - \bm{d_{k}}} F_k =
  \sum\nolimits_{i=1}^{k-1} (F_k \, H_i) F_i - G \, F_k.$
   \end{proof}

   In the following, $\Mac^k_{\bm{d}}$ is not the Macaulay matrix of
   \ref{thm:leadingMonOfMacMatrix}.  It contains less rows because of
   the Koszul F5 criterion.  However, both matrices have the same row
   space, so we use the same name.
   
   \begin{corollary}
     \label{thm:matrixWithCritF5generatesTheSame}
     Using the notation of \Cref{thm:f5Crit}, let
     $\Mac^k_{\bm{d}}$ be a Macaulay matrix of degree $\bm{d}$ wrt
     the order $<$ whose rows are
     {\small
     \begin{align*}
       \rows(\gMac_{\bm{d}}^{k-1})
       \cup
       \left\{
       \monhom{\beta}{d - d_k} \, F_k :
       \!
       \begin{array}{l}
         \monhom{\beta}{d - d_k} \in \KSh_{\bm{d}-\bm d_k} \text{ and }\\
         \monhom{\beta}{d - d_k} \not\in \LM_{<}(\rows(\gMac^{k-1}_{\bm{d} - \bm d_k}))
       \end{array}
       \right\}
     \end{align*}
   }
     \noindent
     The row space of $\Mac^k_{\bm{d}}$ and the  Macaulay matrix of
     \Cref{thm:leadingMonOfMacMatrix} are equal.
   \end{corollary}

   {\footnotesize
   \begin{algorithm}[t]
     \caption{\texttt{ReduceMacaulay}}
     \begin{algorithmic}[1]
     \label{alg:f5Alg}
     \REQUIRE Homogeneous $F_1,\dots,F_k \in \KSh$ of multidegree
     $\bm{d}_1,\dots,\bm{d}_k$, a multidegree $\bm{d}$, and  a
     monomial order $<$.
     \ENSURE The Macaulay matrix of
     $\ideal{F_1,\dots,F_k}_{\bm{d}} \in \KSh$ \wrt $<$ in row echelon
     form.

     \STATE $\Mac^k_{\bm{d}} \leftarrow$ \parbox[t]{175px}{Macaulay
       matrix with columns indexed by the monomials in
       $\KSh_{\bm{d}}$ in decreasing order wrt~$<$}

     \IF{$k > 1$}

     \STATE $\gMac^{k-1}_{\bm{d}} \leftarrow$
     \parbox[t]{175px}{$\texttt{ReduceMacaulay}(\{F_1,\dots,F_{k-1}\}, \bm{d}, <)$}

     \STATE $\gMac^{k-1}_{\bm{d} - \bm d_k} \leftarrow$
     \parbox[t]{176px}{$\texttt{ReduceMacaulay}(\{F_1,\dots,F_{k-1}\},
       \bm{d} - \bm d_k, <)$}

     \FOR{$F \in \rows(\gMac^{k-1}_{\bm{d}})$}

     \STATE Add the polynomial $F$ as a row to $\Mac^k_{\bm{d}}$.

     \ENDFOR

     \ENDIF
     
     \FOR{$\monhom{\alpha}{d - \bm d_k} \in \KSh_{\bm{d}-\bm d_k}
       \setminus \LM_{<}(\rows(\gMac^{k-1}_{\bm{d} - \bm d_k}))
       $}
%
           \STATE Add the polynomial $\monhom{\alpha}{d - \bm d_k} F_k$ as a row to $\Mac^k_{\bm{d}}$.
%
        \ENDFOR

     \STATE $\gMac^{k}_{\bm{d}} \leftarrow$ \texttt{GaussianElimination}($\Mac^k_{\bm{d}}$)      
        
     \RETURN $\gMac^k_{\bm{d}}$
     \end{algorithmic}
   \end{algorithm}
   }

   The correctness of \Cref{alg:f5Alg} follows from \Cref{thm:f5Crit}.
   

   \begin{lemma} \label{thm:H1vanishLastPolInIdeal}
     If $\H_1(F_1,\dots,F_k)_{\bm{d}} = 0$ and there is a syzygy
     $\sum_i G_i F_i = 0$ such that
     $G_i \in \KSh_{\bm{d} - \bm{d}_i}$, then 
     $G_k \in \ideal{F_1,\dots,F_{k-1}}_{\bm{d} - \bm d_k}$.
   \end{lemma}

   \begin{proof}
     We consider the Koszul complex $\Ko(F_1,\dots,F_k)$
     (\Cref{def:KoszulComplex}).  As
     $\sum_i G_i F_i = \delta_1(G_1,\dots,G_k)$,
     the vector of polynomials 
     $(G_1,\dots,G_k)$ belongs to the Kernel of $\delta_1$.  As
     $\H_1(F_1,\dots,F_k)_{\bm{d}}$ vanishes,  the kernel of
     $\delta_1$ is generated by the image of $\delta_2$. The
     latter map is 
     $$(H_{1,2},\dots,H_{k-1,k}) \mapsto \sum\nolimits_{1 \leq i < j \leq k} H_{i,j} (F_j \bm{e}_i - F_i \bm{e}_j),$$
     where $\bm e_i$ and $\bm e_j$ are canonical basis of $\R^k$.
     Hence, there are homogeneous polynomials  $(H_{1,2},\dots,H_{k-1,k})$ such that
     $$(G_1,\dots,G_k) = \sum\nolimits_{1 \leq i < j \leq k} H_{i,j} (F_j \bm e_i - F_i \bm e_j ).$$
     Thus, $G_k = \sum_{i = 1}^{k-1} H_{i,k} F_i$ and so
     $G_k \in \ideal{F_1,\dots,F_{k-1}}_{\bm{d} - \bm d_k}$.
   \end{proof}

   \noindent
   The next lemma shows that we avoid all redundant computations,
   that is all the rows reducing to zero during Gaussian elimination.
   \begin{lemma}
     \label{thm:regAssumpF5noSyz}
     If $\H_1(F_1,\dots,F_k)_{\bm{d}} = 0$, then all the rows of the
     matrix $\Mac^k_{\bm{d}}$ in \Cref{alg:f5Alg} are linearly
     independent.
   \end{lemma}

   \begin{proof}
     By construction, the rows of $\Mac^k_{\bm{d}}$ corresponding to
     $\gMac^{k-1}_{\bm{d}}$ are linearly independent because the
     matrix is in row echelon form. Hence, if there are rows that are
     not linearly independent, then at least one of them corresponds
     to a polynomial of the form  $\monhom{\alpha}{d - \bm d_k} F_k$.
     The right action of the Macaulay matrix $\Mac^k_{\bm{d}}$
     represents a map equivalent to the map $\delta_1$ from the strand
     of Koszul complex \linebreak $\Ko(F_1,\dots,F_k)_{\bm{d}}$.
     Hence, if some of the rows of the matrix are linearly dependent,
     then there is an element in the kernel of
     $\delta_1$. That is, there are
     $G_i \in \KSh_{\bm{d} - \bm d_i}$ such that
     \begin{itemize}[leftmargin=*]
     \item $\sum_{i = 1}^{k-1} G_i F_i$ belongs to the linear span of  $\rows(\gMac^{k-1}_{\bm{d}})$,
     \item the monomials of $G_k$ do not belong to
       $\LM_{<}(\rows(\gMac^{k-1}_{\bm{d} - \bm d_k}))$, and
     \item $\sum_{i = 1}^k G_i F_i = 0$.
     \end{itemize}

     By \Cref{thm:H1vanishLastPolInIdeal},
     $G_k \in \ideal{F_1,\dots,F_{k-1}}_{\bm{d} - \bm d_k}$.
     By
     \Cref{thm:leadingMonOfMacMatrix} and
     \Cref{thm:matrixWithCritF5generatesTheSame},
     the leading
     monomials of $\rows(\gMac^{k-1}_{\bm{d} - \bm d_k})$
     and the ideal $\ideal{F_1,\dots,F_{k-1}}$
     at degree $\bm{d} - \bm{d_k}$
     are the same.
     Hence, we reach a  contradiction
     because we have assumed that the leading monomial of $G_k$ does not
     belong to $\LM_{<}(\rows(\gMac^{k-1}_{\bm{d} - \bm d_k}))$.
   \end{proof}

   \begin{corollary}
     \label{thm:regSkipsEverySyz}
     If $F_1,\dots,F_k$ is a sparse regular polynomial system
     (\Cref{def:Koszul-reg}) and $\bm{d} \in \N^r$ is such that
     $\bm{d} \geq (\sum_i \bm d_i)$, then \linebreak
     {$\texttt{ReduceMacaulay}(F_1, \dots, F_k, \bm{d}, <)$}
     only considers matrices with linearly
     independent rows and avoids all redundant computations.
   \end{corollary}

   To benefit from the Koszul F5 criterion and compute with smaller
   matrices during the \Groebner basis computation we should replace
   Lines 4 -- 8 in \Cref{alg:lazardAlg} by  $\texttt{ReduceMacaulay}(F_1, \dots, F_k, \bm{d}, <)$
   (\Cref{alg:f5Alg}).

\section{\Groebner bases for $0$-dim systems}
\label{sec:solve-I-sat}

We introduce an algorithm, that
takes as input a $0$-dimensional ideal $I$ and computes a \Groebner
basis for the ideal $\left( I : \ideal{\prod_j x_j}^\infty \right)$.
The latter corresponds to the ideal associated to the intersection of
the torus $\T$ with the variety defined by $I$.

Let $f_1,\dots,f_n \in \Kx$ be a square $0$-dimensional system.
First we embed each $f_i$ in $\K[\Z^n]$.
We multiply
each polynomial by an appropriate monomial,
$\mon{\beta_i} \in \K[\Z^n]$, so that $\bm{0}$ is a vertex of
each new polynomial, as well as, a vertex of their 
Minkowski sum.
Let the Newton polytopes be $\polyt_i = \NP(\mon{\beta_i} f_i)$, for $1 \leq i \leq n$,
Let $\polyt_0$ be the standard $n$-simplex; it is the Newton polytope of 
$\NP(1 + \sum_i {x_i})$. We consider the algebras $\KS$ and $\KSh$
associated to the polytopes $\polyt_0,\dots,\polyt_n$ and the embedding
$\mon{\beta_1} f_1,\dots, \mon{\beta_1} f_n \in \KS$. For each $i$, we
consider $F_i \in \KSh_{\bm{e}_i}$ such that
$\dehom(F_i) = \mon{\beta_i} f_i \in \KS$.

\begin{assumption}
  \label{ass:no-sol-at-inf}

  Using the previous notation, let $X$ be the projective toric variety associated to
  $\polyt_0 + \dots + \polyt_n$ (see also the discussion on toric
  varieties at \Cref{sec:reg-inf-sol}).
  Assume that the system $(f_1,\dots,f_n)$ has no
  solutions at infinity \wrt  $X$ (\Cref{def:sol-inf}).
  Further, assume that the system $(f_0, f_1,\dots,f_n)$,
  where $f_0$ is generic linear polynomial, has no solutions over
  $\T$.
\end{assumption}


\begin{lemma}[{\cite[Thm.~3.a]{massri_solving_2016}}]
  \label{thm:KszReg0dimSys}
  Under \Cref{ass:no-sol-at-inf},
  for every $\bm{d} \in \N^{n+1}$ such that
  $\bm{d} \geq \sum_{i > 0} \bm{e}_i$, it holds
  $\H_0(F_1,\dots,F_n)_{\bm{d}} \cong
  \K[\Z^n] / \ideal{f_1,\dots,f_n}.$
\end{lemma}

\begin{lemma}[{\cite[Thm.~3.c]{massri_solving_2016}}]
  \label{thm:AddMonKoszulRegularity}
  Under \Cref{ass:no-sol-at-inf}, for every homogeneous polynomial
  $F_0 \in \KSh_{\bm{d_0}}$ such that the system {\footnotesize
    $(f_1, \dots, f_n,\dehom(F_0))$} has no solutions over $\T$, the
  system $(F_1,\dots,F_n,F_0)$ is Koszul regular
  (\Cref{def:Koszul-reg}) and, for every $\bm{d} \in \N^{n+1}$ such
  that $\bm{d} \geq \sum_i \bm{e}_i + \bm{d_0}$,
  $\ideal{F_1,\dots,F_n,F_0}_{\bm{d}} = \KSh_{\bm{d}}$.
\end{lemma}

\begin{proof}
  The homogenization of system $(f_1,\dots,f_n,\dehom(F_0))$ \wrt the
  toric variety $X$ has no solutions over $X$ (see at
  \Cref{sec:reg-inf-sol} the discussion before \Cref{def:sol-inf}).
  To see this, notice that by  
  \Cref{ass:no-sol-at-inf} 
  the homogenization of the system $(f_1,\dots,f_n)$ \wrt $X$ has no
  solutions over  $X \setminus \T$ (see also \Cref{def:sol-inf}).
  Moreover, 
  we also assume that there are no solutions over $\T$
  of  $(f_1,\dots,f_n,\dehom(F_0))$.

  Now, the proof follows from the argument
  in the proof of {\cite[Thm.~3]{massri_solving_2016}}.
  This argument is the same as in
  \cite[Prop. 3.4.1]{gelfand_discriminants_1994}, where the stably
  twisted condition is given by {\cite[Thm.~1]{massri_solving_2016}}.
 \end{proof}

\begin{corollary}
  \label{thm:monomialMakesKoszulReg}
  For any monomial $\monhom{\alpha}{\text{\unboldmath$D$} \, e_0} \in
  \KSh_{D \, \bm{e}_0}$, the system
  $(F_1,\dots,F_n,\monhom{\alpha}{\text{\unboldmath$D$} \, e_0})$ is
  Koszul regular.
  For every $\bm{d} \in \N^{n+1}$ such that
  $\bm{d} \geq \sum_i \bm{e}_i + D\,\bm{e}_0$, it holds
  $\ideal{F_1 \dots F_n,\monhom{\alpha}{\text{\unboldmath$D$} \,
      e_0}}_{\bm{d}} = \KSh_{\bm{d}}$.
\end{corollary}

\noindent
Fix a graded monomial order $<$ for $\KSh$;
$\L$ is  the set of
monomials that are not leading monomials of
$\ideal{F_1,\dots,F_n}_{\sum_{i \geq 1} \bm{e_i}}$, that is
$$
\L := \left\{ \monhom{\alpha}{\text{\unboldmath$\sum_{i \geq 1}$} \bm{e_i}} \in \KSh_{\sum_{i \geq 1} \bm{e_i}} :
  \begin{array}{r}
    \forall G \in \ideal{F_1,\dots,F_n}_{\sum_{i \geq 1} \bm{e_i}}, \\ \LM_<(G) \neq \monhom{\alpha}{\text{\unboldmath$\sum_{i \geq 1}$} \bm{e_i}}
      \end{array}
\right\}
$$

We will prove that the dehomogenization of these  monomials,
$\dehom(\L)$, forms a monomial basis for
$\K[\Z^n] / \ideal{ f_1,\dots, f_n }$.


\begin{lemma}
  \label{thm:F5BasisIsLinearlyIndependent}
  The monomials in the set  $\dehom(\L)$ are $\K$-linearly independent in
  $\K[\Z^n] / \ideal{f_1,\dots,f_n}$.
\end{lemma}
\begin{proof}
  Assume that the lemma does not hold. Hence, there are
  $c_1,\dots,c_v \in \K$, not all of them $0$,  and $g_1,\dots,g_n \in \K[\Z^n]$ such that
  $\sum_{i} c_i \dehom(\L_i) = \sum_i g_i f_i$. We can clear the
  denominators, introduced by the $g_i$'s,  by choosing a monomial
  $\mon{\alpha} \in \K[\N^n]$ such that, for every $i$,
  $\left( \frac{\mon{\alpha}}{\mon{\beta_i}} \, g_i \right) \in
  \K[\N^n]$.
  Moreover, there is a degree $D \in \N$ and homogeneous polynomials
   $G_i \in \KSh$ of multidegrees
  $(D \, \bm{e}_0 + \sum_{j > 0} \bm{e}_j - \bm{e}_i)$ such that
  $\dehom(G_i) = \left( \frac{\mon{\alpha}}{\mon{\beta_i}} \, g_i
  \right)$ and
  $\monhom{\alpha}{\text{\unboldmath$D$} \, e_0} \sum_{i} c_i \L_i = \sum_i
  G_i \, F_i.$
  By \Cref{thm:AddMonKoszulRegularity},
  $(F_1,\dots,F_n,\monhom{\alpha}{\text{\unboldmath$D$} \, e_0})$ is
  Koszul regular and so, by \Cref{thm:H1vanishLastPolInIdeal}, 
  $\sum_{i} c_i \L_i \in \ideal{F_1,\dots,F_n}_{\sum_{i > 1}
    \bm{e}_i}$. So, a monomial in $\L$ is a leading monomial of
  an element in $\ideal{F_1,\dots,F_n}_{\sum_{i > 1} \bm{e}_i}$.  This
  is a contradiction as, by construction, there is no monomial in
  $\L$ which is a leading monomial of a polynomial in
  $\ideal{F_1,\dots,F_n}_{\sum_{i > 1} \bm{e}_i}$.
\end{proof}

\begin{corollary} \label{thm:F5basisAreMonomialBasis}
  The set of monomials $\dehom(\L)$ is a monomial basis of
  $\K[\Z^n] / \ideal{f_1,\dots,f_n}$.
\end{corollary}

\begin{proof}
  By \Cref{thm:KszReg0dimSys}, the number of elements in the set $\L$
  and the dimension of $\K[\Z^n] / \ideal{f_1,\dots,f_n}$ is the same.
  By \Cref{thm:dehomSameDegIsInjective}, the sets $\L$ and
  $\dehom(\L)$ have the same number of elements.  By
  \Cref{thm:F5BasisIsLinearlyIndependent}, the monomials in the set
  $\dehom(\L)$ are linearly independent.
\end{proof}

\begin{remark}
  \label{rem:compute-L}
  One way to compute the set $\L$ is to compute a basis of the vector
  space $\ideal{F_1,\dots,F_n}_{\sum_{i \geq 1} \bm e_i}$ using
  \Cref{alg:f5Alg}, that is \linebreak
  $\text{\texttt{ReduceMacaulay}}\left((F_1,\dots,F_n),\sum_{i \geq 1}
    \bm{e}_i, < \right)$ .
\end{remark}

For each $F_0 \in \KSh_{\bm{e}_0}$, we will construct a Macaulay
matrix for $(F_1,\dots,F_n,F_0)$ at multidegree
$\bm{1} := \sum_i \bm{e}_i$, say $\Mac(F_0)$;
from this matrix we will recover the
multiplication map of $\dehom(F_0)$ in $\Kx /
\ideal{f_1,\dots,f_n}$.
The rows of $\Mac(F_0)$ are of two kinds: 
\begin{itemize}
\item the polynomials in
  $\rows(\gMac^{n}_{\bm{1}})$, where  \linebreak
  $\gMac^{n}_{\bm{1}}  = \text{\texttt{ReduceMacaulay}}\left((F_1,\dots,F_n),\bm{1}, <\right)$,
\item the polynomials of the form $m \, F_0$, where $m \in \L$.
\end{itemize}

\begin{lemma} \label{thm:matrixF0}
  The matrix $\Mac(F_0)$ is always square. It is full-rank if and
  only if $(F_1,\dots,F_N, F_0)$ is Koszul regular.
\end{lemma}

\begin{proof}
  According to the Koszul F5 criterion (see \Cref{thm:f5Crit}), the
  row space spanned by $\Mac(F_0)$ is the same as the vector space
  \linebreak
  $\ideal{F_1,\dots,F_n,F_0}_{\bm{1}}$
  for any choice of $F_0$.
  We can consider an $F_0$ such that $(F_1,\dots,F_n,F_0)$
  is Koszul regular, by \Cref{thm:monomialMakesKoszulReg}.
  Then, the
  rows of $\Mac(F_0)$ generate $\KSh_{\bm{1}}$ and, by
  \Cref{thm:regAssumpF5noSyz}, the rows of $\Mac(F_0)$ are linearly
  independent. Hence, by \Cref{thm:AddMonKoszulRegularity}, for this
  particular $F_0$, the matrix $\Mac(F_0)$ is square and full-rank.
  However, the matrix $\Mac(F_0)$ is square for any choice of $F_0$,
  because its number of rows does not depend on $F_0$.
  Nevertheless, it is not full-rank for any choice of $F_0$.
  If $\Mac(F_0)$ is full-rank, then $(F_1,\dots,F_n,F_0)$ is Koszul
  regular because, by the sparse Nullstellensatz
  \cite[Thm.~2]{sombra_sparse_1999}, the homogenization of the system
  $(f_1,\dots,f_n,\dehom(F_0))$ has no solutions over
  $\T$. Consequently, the proof follows from
  \Cref{thm:AddMonKoszulRegularity}.
\end{proof}



We reorder the columns of $\Mac(F_0)$ as shown in \Cref{eq:reorderMF0}, such that
\begin{itemize}
\item the columns of the submatrix 
  $
  \left(
    \begin{smallmatrix}
      M_{1,2}(F_0) \\ \hline
      M_{2,2}(F_0) 
    \end{smallmatrix}
  \right)
  $
  correspond to  monomials of the form
   $m \, \monhom{0}{\bm{e}_0}$, where $m \in \L$, and 
 \item the rows of
   $
\left(
\begin{smallmatrix}
  M_{2,1}(F_0) & |  & M_{2,2}(F_0) 
\end{smallmatrix}
\right) $ are polynomials of the form $m \, F_0$, where $m \in \L$.
\end{itemize}
\vspace{-12pt}
{\footnotesize
 \begin{align}
   \label{eq:reorderMF0}
   \begin{tabular}{r}
     $\rows(\gMac^{n}_{\bm{1}})
    \left\{ 
     \lefteqn{\phantom{
     \begin{array}{c}  \\[5pt] M_{1,1}(F_0)  \\[10pt] \hline \end{array}
     }}
     \right.$\\
     $F_0 \cdot \L
     \left\{\lefteqn{\phantom{
     \begin{array}{c}  \\[5pt]  M_{1,1}(F_0) \\[10pt] \hline \end{array}
     }}\right.$
   \end{tabular}
  \left[\phantom{\begin{matrix}a_0\\ \ddots\\a_0\\b_0\\ \ddots\\b_0 \end{matrix}}\right.\hspace{-1.5em}
  \begin{array}{c |}
    \\[5pt]
    \hspace{8pt} M_{1,1}(F_0) \hspace{8pt} \\[10pt] \hline\hline
    \\[5pt]
    M_{2,1}(F_0)         \\[10pt]
  \end{array}
  \overbrace{
    \begin{array}{| c}
      \\[5pt]
        \hspace{2pt} M_{1,2}(F_0) \\[10pt] \hline\hline
      \\[5pt]
      M_{2,2}(F_0) \\[10pt]
    \end{array}
  }^{\monhom{0}{1} \cdot \L}
  \hspace{-1.5em}
  \left.\phantom{\begin{matrix}a_0\\ \ddots\\a_0\\b_0\\ \ddots\\b_0 \end{matrix}}\right]\hspace{-1em}
 \end{align}
}

\noindent
We prove that $M_{1,1}(F_0)$ is invertible and the Schur
complement of $M_{2,2}(F_0)$, 
$
M_{2,2}^c(F_0) := (M_{2,2} - M_{2,1} \, M_{1,1}^{-1} \, M_{1,2})(F_0),
$
is the multiplication map of $\dehom(F_0)$ in the  basis
$\dehom(\L)$ of $\K[\Z^n] / \ideal{f_1,\dots,f_n}$.

\begin{lemma}
  If $(F_1,\dots,F_n,\monhom{0}{e_0})$ is Koszul regular then, for any
  $F_0 \in \KSh_{\bm{e}_0}$, the matrix $M_{1,1}(F_0)$ is invertible.
\end{lemma}

\begin{proof}
  By \Cref{thm:matrixF0}, as the system $(F_1,\dots,F_n,\monhom{0}{e_0})$
  is Koszul regular, then  the matrix $\Mac(\monhom{0}{e_0})$
  is invertible.  As \linebreak $M_{2,1}(\monhom{0}{e_0})$ is the zero matrix and
  $M_{2,2}(\monhom{0}{e_0})$ is the identity, then
  $M_{1,1}(\monhom{0}{e_0})$ must be invertible. By construction, the
  matrices $M_{1,1}(F_0)$ and $M_{1,2}(F_0)$ are independent of the
  choice of $F_0$. Hence, for any $F_0$  the matrix $M_{1,1}(F_0)$ is invertible.
\end{proof}

\begin{theorem}
  The multiplication map of $\dehom(F_0)$ in the monomial basis
  $\dehom(\L)$ of $\K[\Z^n] / \ideal{f_1,\dots,f_n}$
  is the Schur complement of $M_{2,2}(F_0)$, that is
  $
  M_{2,2}^c(F_0) := (M_{2,2} - M_{2,1} \, M_{1,1}^{-1} \, M_{1,2})(F_0).
  $
\end{theorem}

\begin{proof}
  Note that for every $F_0 \in \KSh_{\bm{e}_0}$ and each element $L_i$ of $\L$,
  $L_i F_0 \equiv
  \monhom{0}{e_0} \sum_{j} (M_{2,2}^c(F_0))_{i,j}  L_j$
  in $\KSh / \ideal{F_1,\dots,F_n}$,
  where $(M_{2,2}^c(F_0)_{i,j}$ is the $(i,j)$ element of the matrix $M_{2,2}^c(F_0))$.
  Hence, if we dehomogenize
  this relation we obtain that,
  $\dehom(L_i) \dehom(F_0) \equiv \sum_{j} (M_{2,2}^c(F_0))_{i,j}
  \dehom(L_j)$ in
  $\KS / \ideal{\mon{\beta_1} f_1,\dots,\mon{\beta_n} f_n}$.
  As
  $\KS \subset \K[\Z^n]$, the same relation holds in
  $\K[\Z^n] / \ideal{f_1,\dots,f_n}$.
  By
  \Cref{thm:F5basisAreMonomialBasis}, the set $\dehom(\L)$ is a monomial basis
  of $\K[\Z^n] / \ideal{f_1,\dots,f_n}$.
  Therefore, $M_{2,2}^c(F_0)$
  is the multiplication map of $F_0$ in
  $\K[\Z^n] / \ideal{f_1,\dots,f_n}$.
\end{proof}

Using the multiplication maps in $\K[\Z^n] / \ideal{f_1,\dots,f_n}$
and the FGLM algorithm \cite{faugere_efficient_1993}, we can compute a
\Groebner basis for \linebreak
$\ideal{f_1,\dots,f_n} : \ideal{\prod_{i} x_i}^\infty$ over $\Kx$.
The latter is the saturation over $\K[\N^n]$ of the ideal
$\ideal{f_1,\dots,f_m}$ by the product of all  the variables.

\begin{lemma}
  Consider polynomials $f_1,\dots,f_m \subset \K[\N^n]$
  such their ideal over $\K[\Z^n]$,
  $\ideal{f_1,\dots,f_m}_{\K[\Z^n]}$, is 0-dimensional.
  Let \linebreak $\ideal{f_1,\dots,f_m}_{\K[\N^n]}$ be the ideal
  generated by $f_1,\dots,f_m$ over $\K[\N^n]$. Then, the sets
  $\ideal{f_1,\dots,f_m}_{\K[\Z^n]} \cap \K[\N^n]$ and
  $\ideal{f_1,\dots,f_m}_{\K[\N^n]} : \ideal{\prod_i x_i}^\infty$ are
  the same.  The latter is an ideal over $\K[\N^n]$.
\end{lemma}

\begin{proof}
  Consider $f \in \ideal{f_1,\dots,f_m}_{\K[\Z^n]}$. Then there are
  $g_i \in \K[\Z^n]$ such that $f = \sum_i g_i f_i$. We can clear the
  denominators introduced by the $g_i$'s by multiplying both sides by a monomial
  $(\prod_j x_j)^d$, where $d$ is big enough.
  Then,
  $(\prod_j x_j)^d \, f = \sum_i ((\prod_j x_j)^d \, g_i) f_i$ and
  $((\prod_j x_j)^d \, g_i) \in \K[\N^n]$. Thus,
  $(\prod_j x_j)^d \, f \in \ideal{f_1,\dots,f_m}$ and
  $f \in \ideal{f_1,\dots,f_m} : \ideal{\prod_j x_j}^\infty$.  The
  opposite direction is straightforward  as $\prod_i x_i$ is a unit in
  $\K[\Z^n]$.
\end{proof}

We can perform FGLM over $\K[\Z^n]$ to recover a \Groebner basis for
$\ideal{f_1,\dots,f_m}_{\K[\N^n]} : \ideal{\prod_i x_i}^\infty$ by
considering the multiplication maps of each $x_i$.
These correspond to 
$M_{2,2}^c(\monhom{\alpha_i}{e_0})$, where $\bm{\alpha_i}$ is such
  that $\dehom(\monhom{\alpha_i}{e_0}) = x_i$.
We skip the details of this procedure.

\subsection{Complexity}
  \label{sec:complexity}
  
  We estimate the arithmetic complexity of the  algorithm  in
  Sec~\ref{sec:solve-I-sat}; it is polynomial  \wrt the
  Minkowski sum of the polytopes.
  We omit the cost of computing all the monomials in $\KSh_{\bm{d}}$
  and we only consider the complexity to read them.
  Our purpose is to highlight the dependency on the Newton polytopes.
  A more detailed analysis might give sharper bounds.
  \begin{definition}
    For polytopes $\polyt_0,\dots,\polyt_n$ and for each
    multidegree $\bm{d} \in \N^{n+1}$ of $\KSh$, let  $P(\bm{d})$ be 
    the number of integer points in the \mksum of the
    polytopes given by $\bm{d}$,

    \vspace{\abovedisplayshortskip}
    \hfil
    $P(\bm{d}) =
    \# \Big( ( \sum\nolimits_{j = 0}^n d_j \Delta_j) \cap \Z^{n} \Big).$
  \end{definition}

  \noindent
  Note that $P(\bm{d})$ equals  the number of
  different monomials in $\KSh_{\bm{d}}$.
  
  \begin{lemma} \label{thm:complexityMacMatrix}
    Let $F_1,\dots,F_k \in \KSh$ be a (sparse) regular sequence and
    let $\bm{d}_i \in \N^{n+1}$ be the multidegree of $F_i$, for
    $i \in [k]$. Consider a multigraded monomial order $<$. For
    every multidegree $\bm{d} \in \N^{n+1}$ such that
    $\bm{d} \geq \sum_i \bm{d}_i$, the arithmetic complexity
    of computing
    $\text{\texttt{ReduceMacaulay}}\left((F_1,\dots,F_k),\bm{d}, <
    \right)$ is 
    $O( 2^{k+1} \,P(\bm{d})^{\omega} )$,
    where $\omega$ is the constant of matrix multiplication.
  \end{lemma}

  \begin{proof}
    By \Cref{thm:regSkipsEverySyz}, as $F_1,\dots,F_k \in \KSh$ is a
    (sparse) regular sequence and
    $\bm{d} \geq \sum_i \bm{d}_i$, all the matrices that
    appear during the computations of 
    $\text{\texttt{ReduceMacaulay}}\left((F_1,\dots,F_k),\bm{d}, <
    \right)$
    are full-rank and their rows are linearly independent.
    Hence, their number of rows is at most
    their number of columns. 
    The number of columns of a Macaulay matrix of multidegree
    $\bm{d}$ is $P(\bm{d})$.
    Thus, in this case, the complexity of Gaussian elimination is
    $O(P(\bm{d})^\omega)$ \cite{gg-mca-13}.
    If  $C(k,\bm{d})$ is cost of 
    $\text{\texttt{ReduceMacaulay}}\left((F_1,\dots,F_k),\bm{d}, <
    \right)$, then we have the following recursive relation
    $$C(k,\bm{d}) =
    \left\{
    \begin{array}{c c}
      O( P(\bm{d})^\omega ) & \text{if $k = 1$, } \\
      C(k-1,\bm{d}) + C(k-1,\bm{d} - \bm{d}_k) + O( P(\bm{d})^\omega ) & \text{if $k > 1$.}
    \end{array}
    \right.
    $$
    The cost $C(k-1,\bm{d})$ is greater than $C(k-1,\bm{d} - \bm{d}_k)$,
    as it involves bigger matrices.
    Hence, we obtain
    $C(k,\bm{d}) \!= \!O( 2^{k+1} P(\bm{d})^{\omega}).$ \qedhere
  \end{proof}
      
  \begin{theorem}
    Consider an affine polynomial system $(f_1,\dots,f_n)$ in $\Kx$
    such that \Cref{ass:no-sol-at-inf} holds and the system
    $(F_1,\dots,F_n)$ is (sparse) regular, where
    $F_i \in \KSh_{\bm{e_i}}$ and $\dehom(F_i) = f_i$, for $i \in [n]$.  Then, the
    complexity of computing
    $\ideal{f_1,\dots,f_n} : \ideal{\prod_i x_i}^\infty$ is
    $$
    O( 2^{n+1} \, P(\bm{d})^{\omega} + n \, \MV(\polyt_1,\dots,\polyt_n)^3).      
    $$
    
  \end{theorem}
  \begin{proof}
    We need to compute:
    \begin{itemize}[leftmargin=*]
    \item The set
      $\rows(\text{\texttt{ReduceMacaulay}}\left((F_1,\dots,F_n),\sum_{i
          > 1} \bm{e_i}, < \right))$ to generate $\L$ (\Cref{rem:compute-L}). By
      \Cref{thm:complexityMacMatrix}, this costs
      $O( 2^{n+1} \,P({ \sum_{i > 1} \bm{e_i} })^{\omega} )$.
    \item The set
      $\rows(\text{\texttt{ReduceMacaulay}}\left((F_1,\dots,F_n),\bm{1},
        < \right))$ to generate the matrix $\Mac(F_0)$ of
      \Cref{thm:matrixF0}.  By \Cref{thm:complexityMacMatrix}, it
      costs $O( 2^{n+1} \, P(\bm{1})^{\omega} )$.
    \item For each variable $x_i$, the Schur complement of
      $\Mac(F_0)$, for $\dehom(F_0) = x_i$.
      The cost of each Schur complement computation is
      $O( P(\bm{1})^{\omega} )$, and so the cost of this step is
      $O( n \, P(\bm{1})^{\omega} )$.
      
    \item The complexity of FGLM depends on the number of solutions,
      and in this case it is $O(n \, \MV(\polyt_1,\dots,\polyt_n)^3)$
      \cite{faugere_efficient_1993}. \hfill $\qedhere$
    \end{itemize}
  \end{proof}

  Note that $\MV(\polyt_1,\dots,\polyt_n) < P(\bm{1})$. Hence, to
  improve the previous bound for lexicographical orders, we can follow
  \cite{faugere_polynomial_2013}.



{\footnotesize
  \noindent
\textbf{Acknowledgements:}   We thank C. D'Andrea, C. Massri, B. Mourrain,
P.-J. Spaenlehauer, and B. Teissier for the helpful
discussions and references.
}



\bibliographystyle{ACM-Reference-Format}
\bibliography{biblio}

\newpage

\appendix

\section{Counter-example to the complexity bounds in \cite{faugere_sparse_2014}}
\label{sec:counter-example}

Let $\polyt$ be the standard 2-simplex and consider the regular system
given by two polynomials $F_1,F_2 \in \KSh_2$ of degree $2$,
{\footnotesize
\begin{align*}
F_1 := & \monhomNoBM{(2,0)}{2}+\monhomNoBM{(1,1)}{2}+\monhomNoBM{(0,2)}{2}+\monhomNoBM{(1,0)}{2}+\monhomNoBM{(0,1)}{2}+\monhomNoBM{(0,0)}{2} \\
F_2 := & \monhomNoBM{(2,0)}{2}+2 \,
\monhomNoBM{(1,1)}{2}+3\,\monhomNoBM{(0,2)}{2}+4\,\monhomNoBM{(1,0)}{2}+5\,\monhomNoBM{(0,1)}{2}+6\,\monhomNoBM{(0,0)}{2}
\end{align*}
}
Consider the graded monomial order $<$ given by
{\footnotesize
\begin{equation*}
  \monhomNoBM{(x_1,y_1)}{d_1} < \monhomNoBM{(x_2,y_2)}{d_2} 
  \iff
  \begin{cases}
    d_1 < d_2,  \text{ or} \\
    d_1 = d_2 \, \text{ and } \, x_1 < x_2, \text{ or} \\
    d_1 = d_2 \, \text{ and } \, x_1 = x_2 \, \text{ and } \, y_1 < y_2
  \end{cases} .
\end{equation*}
}
In this case, the bound in \cite[Lem.~5.2]{faugere_sparse_2014} is
$3$, meanwhile the maximal degree of an element in the \Groebner basis
of $(F_1,F_2)$ \wrt $<$ has degree $4$.

\section{Algorithm to compute \Groebner basis over the standard algebra}

  \begin{algorithm}
    \caption{\texttt{compute-0-Dim-GB}}
     \begin{algorithmic}[1]
     \label{alg:solving}
     \REQUIRE Affine system $f_1,\dots,f_n \in \Kx$ and a monomial
     order $<$ for $\Kx$, such that it has a
     finite number of solutions over $\T$ and satisfies
     \Cref{ass:no-sol-at-inf}.
     %
     \ENSURE \Groebner basis $G$ for the ideal
     $\ideal{f_1,\dots,f_n} : \ideal{\prod_i x_i}^\infty$ \wrt the
     monomial order $<$.
     \newline

     \STATE Consider the semigroup algebra $\KSh$ related to the
     polytopes of $f_1, \dots,f_n$ and the standard n-simplex.

     \STATE Choose a multigraded monomial order $\bm{\<}$ for $\KSh$.

     \STATE For each $i \in [n]$, choose $F_i \in \KSh_{\bm{e_i}}$
     such that $\dehom(F_i) = f_i$.

     \STATE
     $C \leftarrow
     \rows(\text{\texttt{ReduceMacaulay}}\left((F_1,\dots,F_n),{\sum_{i
           > 0} \bm{e}_i}, \bm{\<} \right))$,

     \STATE
     $P \leftarrow
     \rows(\text{\texttt{ReduceMacaulay}}\left((F_1,\dots,F_n),\bm{1}, \bm{\<} \right))$,

     \STATE
     $\L \leftarrow \left\{ \monhom{\alpha}{\sum_{i \geq 1} \bm{e_i}}
       \in \KSh_{\sum_{i \geq 1} \bm{e_i}} :
       \monhom{\alpha}{\sum_{i \geq 1} \bm{e_i}} \not\in C
     \right\}$

     \FORALL {$x_i \in \Kx$}

     \STATE Choose a monomial $m \in \KSh_{\bm{e_0}}$ such that
     $\dehom(m) = x_i$.

     \STATE $\Mac(m) \leftarrow$ Macaulay
     matrix of degree $\bm{1}$ \wrt $\bm{\<}$.

     \FORALL {$F \in P$}
     \STATE Add $F$ to $\Mac(m)$
     \ENDFOR

     \FORALL {$L_i \in \L$}
     \STATE Add $L_i \, m$ to $\Mac(m)$
     \ENDFOR

     \STATE Rearrange $\Mac(m)$ as
     $\left[\begin{smallmatrix} M_{1,1}(m) & M_{1,2}(m) \\ M_{2,1}(m)
         & M_{2,2}(m) \end{smallmatrix}\right]$ like
     (\ref{eq:reorderMF0}).

     \STATE $M_{x_i} \leftarrow (M_{2,2} - M_{2,1} \, M_{1,1}^{-1} \, M_{1,2})(m)$.

     \ENDFOR

     \STATE Perform FGLM with the multiplication matrices
     $M_{x_1},\dots,M_{x_n}$ \wrt $>$ to obtain \Groebner basis $G$.

     \RETURN $G$.
     \end{algorithmic}
   \end{algorithm}


\end{document}